\documentclass{sigcomm-alternate}

\usepackage{cite}
\usepackage{url}
\usepackage{graphicx}
\usepackage{subfig}
\usepackage{color}

\newtheorem{mythm}{Theorem}

\title{You Really Need A Good Ruler to Measure Caching Performance in Information-Centric Networks}

\author{Liang Wang, Jussi Kangasharju$^*$, Jon Crowcroft\\ \\University of Cambridge, UK \qquad University of Helsinki, Finland$^*$}

\begin{document}
\maketitle

\begin{abstract}
  Information-centric networks are an interesting new paradigm for distributing content on the Internet.
  They bring up many research challenges, such as addressing content by name, securing content, and wide-spread caching of content.
  Caching has caught a lot of attention in the research community, but a lot of the work suffers from a poor understanding of the different metrics with which caching performance can be measured.
  In this paper we not only present a comprehensive overview of different caching metrics that have been proposed for information-centric networks, but also propose the coupling factor as a new metric to capture the relationship between content popularity and network topology.
  As we show, many commonly used metrics have several failure modes which are largely ignored in literature.
  We identify these problems and propose remedies and new metrics to address these failures.
  Our work highlights the fundamental differences between information-centric caches and ``traditional'' cache networks and we demonstrate the need for a systematic understanding of the metrics for information-centric caching.
  We also discuss how experimental work should be done when evaluating networks of caches.
\end{abstract}

\section{Introduction}
\label{sec:intro}



Information-centric networking (ICN) provides a new paradigm for addressing and accessing content on the Internet.
The current Internet was developed as a host-centric network, where the main focus was on interconnecting computers, or hosts, but the modern usage of Internet is very much information-centric, i.e., users do not care from where the information they want to access comes from; they simply are interested in getting the information.
The web is in its essence a host-centric system, although content delivery networks (CDN) and technologies do break the dependence on (particular) hosts serving specific content to some extent.
However, fundamentally the web is still a host-centric system and its different components, such as naming and security, are tied to this host-centric world.

ICN puts the focus on the content as opposed to the hosts to address the architectural issues preventing the web from becoming a full-blown information-centric system.
There are several independent proposals around ICN~\cite{jacobson:ccn, koponen:dona, netinf, psirp}.
They each present a different solution to try to re-construct the current Internet, and build a new architecture around the notion of content.
While the details in the proposal differ, we can identify three common components that are fundamental to information-centric networking: addressing content by name, securing content, and wide-spread caching.

Of these three, the last one, wide-spread caching, seems to have attracted the most attention in the research world lately.
Our focus in this paper is on measuring the effectiveness of caching, but first we outline the main issues in naming and security, highlighting in particular how they impact caching.

Addressing content by name is an important change to how content is addressed in the web.
Although URLs are ``names'' of content, they have internal structure which indicates the server hosting that content as well as a local ``path'' on the server to the content.
While at first sight similar, names in ICN may have structure, but the structure does not identify a particular server in the network that would need to be accessed to retrieve the content.

Content discovery is a big challenge in ICN and two different choices seem to emerge from the ICN proposals.
One possibility is to use an indexing service~\cite{netinf} which keeps track of copies of objects; however it is not clear if this will scale up to a global scale.
The other possibility, used in most of the other ICN proposals, is to route requests based on some components in the content name, with the hope that this routing converges on a copy of the object.
Especially in the latter case (sometimes combined with well controlled scoped-flooding\cite{wang:diluvian}), en-route caching becomes an attractive option to speed up discovery and spread the load on content distribution.
In this paper, we mainly follow this kind of a model and assume that content requests are routed in the same way from all over the network and that the routing converges towards existing copies of the content.
Caching is assumed to happen en-route and cached copies are not tracked in any way.

Security on the web is essentially based on identifying the server providing the content via SSL and its associated certificates.
Since content no longer has a single origin in ICN, this approach does not work anymore.
Instead, the ICN approaches all focus on securing the content, by ensuring via signatures and public keys that the content has not been tampered with in the network.
This also allows caching to take place since any piece of content, no matter from which server it is served from, can be authenticated to be the same content that the originator put in the network.
Obviously, this does not tie the content to a real world entity, but this can be achieved in a similar way to how it is done on the web.

Finally, wide-spread caching is used to store the content in the network and allow for faster delivery.
Caching also reduces traffic in the network and is therefore attractive for network operators since it has the potential to reduce their costs.

Although all three above factors are fundamental to ICN, caching seems to have attracted the most attention in recent research.
Caching is a topic that has been researched in many different contexts and it is attractive in the sense that it can be measured quantitatively with relative ease, whereas effectiveness of naming schemes or security solutions tend towards more qualitative measurements.
However, a lot of the work on caching in ICN uses the ``old'' caching metrics that are known from processors or web caches.
As we discuss in this paper, ICN is a network of caches and there are fundamental differences between ICN caching and, e.g., web caching.
Web caches can also be organised in networks, however they work in a fundamentally different way from the network of caches in ICN.

Our key contributions in this paper are to highlight the fundamental differences between different caching metrics and show how they impact the metrics that should be used to measure the effectiveness of caching.
We specifically propose using the coupling factor as a new ICN metric to capture the relationship between content popularity and network topology.
We present several metrics and show how they vary in their complexity and expressiveness by using a dimensionality notion to illustrate where hits occur in a network.
The goal of this paper is to demonstrate that caching in ICN is a novel area of research and that existing measurement solutions have only limited applicability in this field.

\section{System View}
\label{sec:system-view}

In this paper we focus on CCN-like~\cite{jacobson:ccn} ICN where requests for (pieces of) content are forwarded via routers and these routers are equipped with a cache where they can store content.
We focus on the case of a single ISP, as shown in Figure~\ref{fig:network-model}, which depicts several clients, one server, and a network of routers.
Some of the routers are connected towards clients and some towards servers.
We do not distinguish whether these are actual clients or other ISPs that are connected to the clients; they represent the incoming requests.
Likewise, the servers represent sources of content and do not necessarily need to be connected to this particular ISP.
From the ISP's point of view, the traffic reduction that caching brings can have two goals.
Firstly, it reduces traffic towards the servers (\emph{inter-ISP traffic}) which typically has a (direct) financial cost for the ISP.
Secondly, it reduces traffic within the ISP's own network (\emph{intra-ISP traffic}).
Intra-ISP traffic has no direct connection to the ISP's costs, however lower intra-ISP traffic means that the ISP is able to serve more customers with the same infrastructure, which does have a positive financial effect.

\begin{figure}[!tb]
  \centering
  \includegraphics[width=8.5cm]{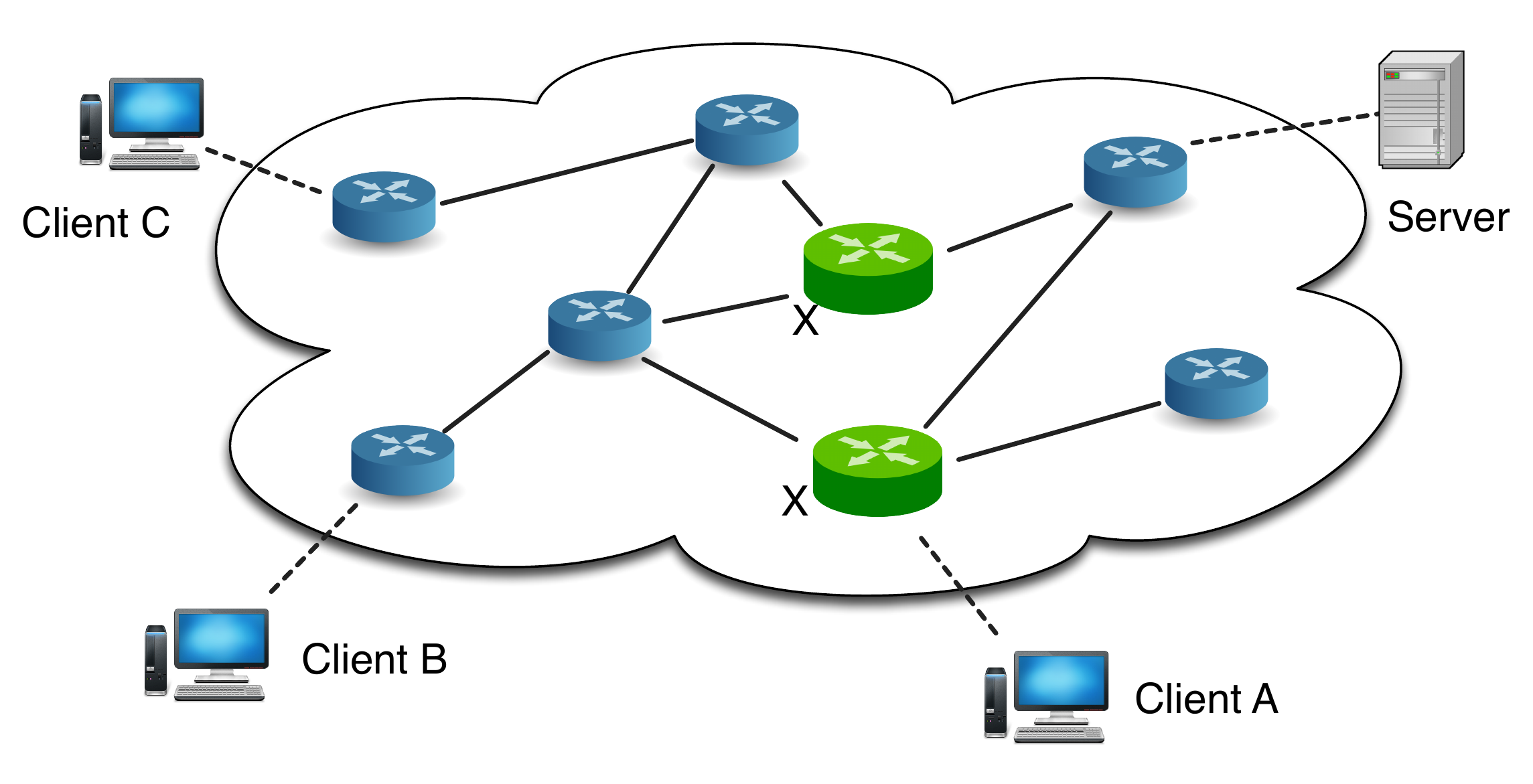}
  \caption{Model of the network}
  \label{fig:network-model}
\end{figure}

The fundamental difference between ICN caching and, for example, web caching is that in ICN caching we have a \emph{network of caches} that can work together to optimise the performance of the whole network.
Although hierarchical web caching implemented similar networks of caches, each cache was operated by a different real-world entity and attempted to optimise its own performance.
This led to issues like the filtering effect~\cite{Wolman:1999:SPC:319151.319153,Williamson:2002:FEW} where first caches in the hierarchy capture the most popular objects because they attempt to optimise their own performance.
This in turn leads to the following caches to see a request stream with less locality, making their performance suffer.
In the context of web caching where every cache is operated by a different entity, this is reasonable, but in the context of ICN, a single entity controls multiple caches and is able to make them cooperate to optimise the overall performance.
Results in~\cite{wong:globecom2012} show that even a simple randomization of where to cache a particular piece of content has a significant boost of overall performance because it mitigates the filtering effect.
As example, consider Figure~\ref{fig:network-model}.
In the web caching model, the routers next to the clients would cache the most popular content, but in ICN, it is feasible to have some other routers do that, for example the green routers marked with X.
In other words,~\cite{wong:globecom2012} shows that caches being less greedy in optimising their own performance is beneficial to the whole system.
This implies that when evaluating networks of ICN caches, \emph{we need to look at the performance of the whole network} instead of optimising performance of individual caches.

\section{Metrics}
\label{sec:metrics}

We now present the main contribution of the paper and outline different metrics that can be used to measure performance of a network of caches.
We consider three metrics that have been used in literature and present a new metric called \emph{coupling factor}.

\begin{figure*}[!tb]
  \centering
  \subfloat[1-dimension metric, BHR only tells you that \textit{how many} hits happen in the \textit{whole} network. The whole network is treated as a single entity.]{\includegraphics[width=5.5cm]{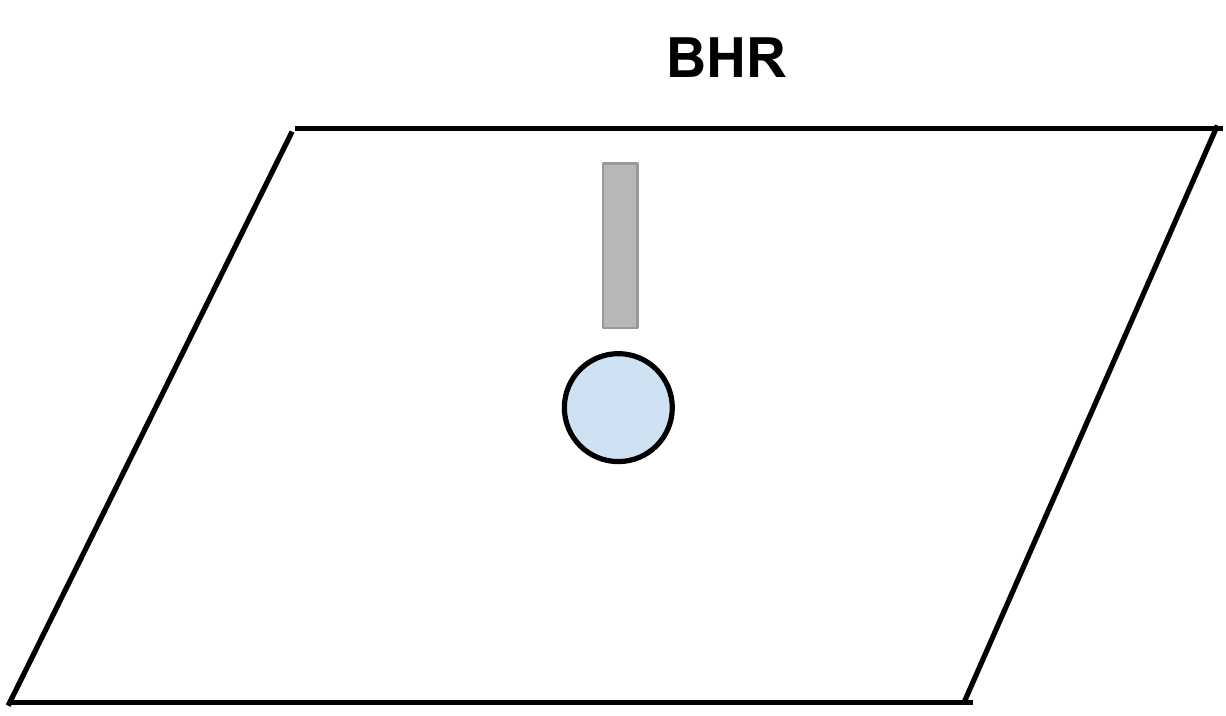}}
  \quad
  \subfloat[2-dimension metric, FPR tells you both \textit{how many} hits and \textit{where} they happen along a path (source $\leftrightarrow$ destination) on average in the network.]{\includegraphics[width=5.5cm]{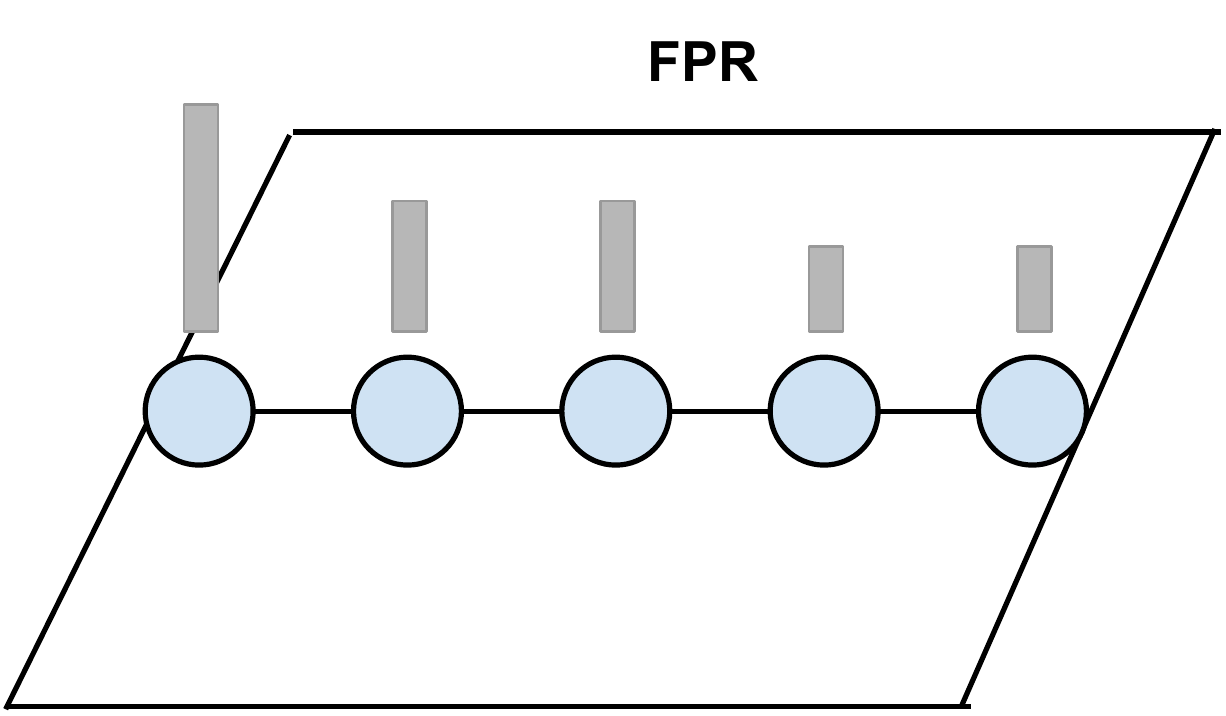}}
  \quad
  \subfloat[3-dimension metric, CPF tells you \textit{how many} hits and \textit{where} they happen on average in the network (e.g., network edge or network core).]{\includegraphics[width=5.5cm]{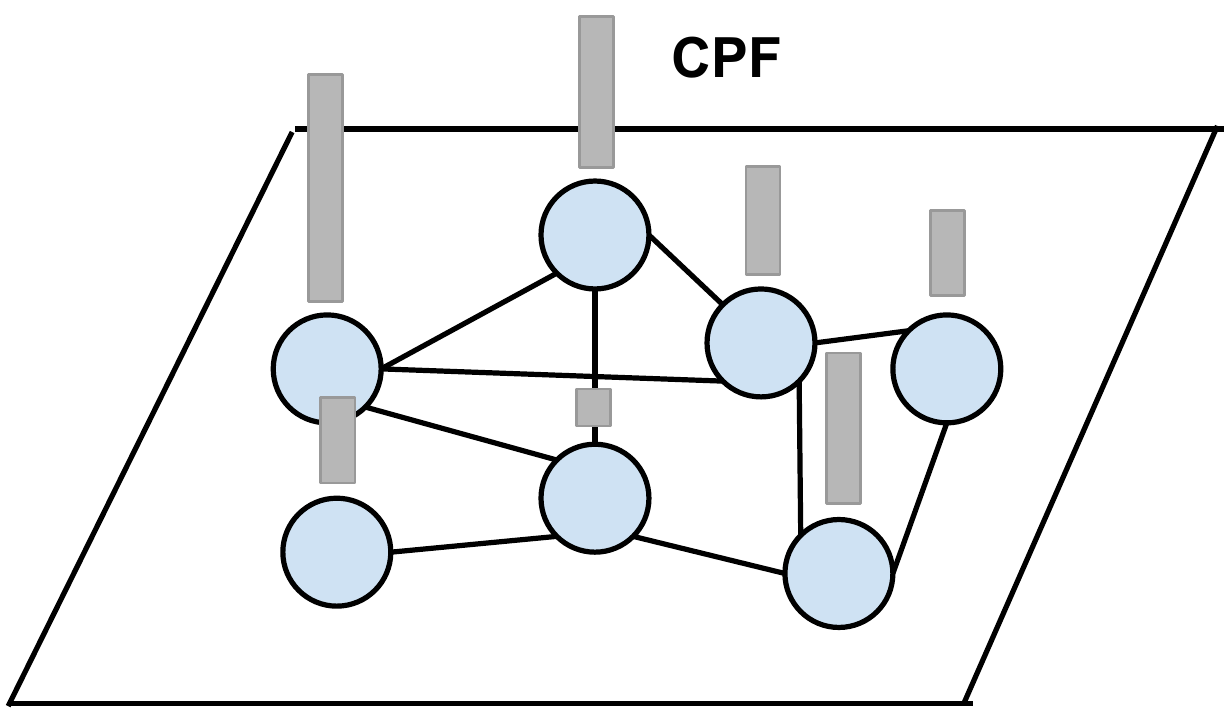}}
  \caption{An illustration of the information contained in three measurement metrics: (Byte) Hit Rate (BHR), Footprint Reduction (FPR), and Coupling Factor (CPF). The calculation becomes more complicated as the information increases.
}
  \label{fig:coupling}
\end{figure*}

\subsection{(Byte) Hit Rate}
\label{sec:hit-rate}

Typically cache performance has been measured via hit rate, which captures the ratio between cache hits (requests found in the cache) to the total number of requests seen by the cache.
Byte hit rate is its natural extension where every hit is weighted by the size of the object, hence byte hit rate measures the reduction in outgoing traffic from the cache.
As our focus is on traffic reduction, we use byte hit rate in the following.
When we apply byte hit rate as the performance metric, we effectively aggregate the whole network of caches as a single cache and look at its performance.
(Note that since multiple caches may hold a copy of the same object in ICN, such an ``aggregate'' cache has less storage than the individual caches together; this does not influence the metric.)
Byte hit rate is location-agnostic since it only cares whether there was a hit in any cache; it does not provide any information about where the hit happened.

Byte hit rate is an often-used metric, partly because caches have traditionally been measured by hit rate, partly because it is easy to compute, and partly because it translates directly to savings in inter-ISP traffic, i.e., financial savings.
Reducing duplicate copies of objects in the network is the most effective way of improving byte hit rate; however an efficient reduction in number of copies requires an efficient cooperation method between the caches to discover the cached copies. Typical examples are various distribute key-value stores \cite{maymounkov:kademlia, ratnasamy:can} which is able to combine distributed caches into a big logical cache via collaboration.

While byte hit rate is easy to compute, it treats the network of ICN caches as a black box since it does not take into account where the hit happens.
Another argument against byte hit rate as a metric stems from the current trends of content distribution in the Internet.
Large content delivery networks or content providers, who host the most popular content, install their servers in or close to the ISPs where the users are.
Although the servers are in the ISP's network, the normal way of calculating byte hit rate would consider them external, thus traffic to them would be counted the same way as any outgoing traffic; yet there is typically no cost to the ISP for traffic to them.
Hence, most of the popular content actually comes from inside the ISP, and only the savings in the less popular content are relevant for the ISP.
Byte hit rate is not able to capture this and therefore we recommend that it should not be used as a general metric; in specific situations it may be appropriate, but it is not appropriate as a general metric for all situations.

\subsection{Average Hops}
\label{sec:average-hops}

Measuring the number of hops a request needs to traverse in order to find the content is also a metric that has been recently used~\cite{seyedicn2013}.
It is appealing in the sense that it augments byte hit rate by taking into account where the hit happens, however it does not provide any meaningful way of estimating savings in outgoing traffic.
In addition, as is done in~\cite{seyedicn2013}, average hops is sometimes used as a proxy for download latency.

Our previous work~\cite{wong:globecom2012} shows that average hops as a metric does not discriminate well, i.e., while the qualitative ranking of caching solutions is correct, the quantitative differences between them are very small, which can easily lead to an impression that the performance differences would be small~\cite{seyedicn2013}.
Other metrics we consider in this paper do not share this weakness.
The reason behind this is that average hops measures absolute values and because many networks are scale-free, the number of hops is typically small.
Hence, differences between caching strategies will appear small, but this is actually an artifact of the metric, not an indicator that the strategies would be close to each other according to other metrics.

Another difficulty in using average hops as a metric relates to what value to assign to content retrieved from outside the ISP, i.e., a miss.
Assigning a high value puts emphasis on avoiding misses, i.e., the metric becomes similar to hit rate.
Assigning a low value emphasizes the location of content in the ISP's network, i.e., it gives an impression of intra-ISP traffic.
However, as the amount of data is not part of the metric, it does a poor job in capturing something useful and a better metric, like footprint reduction described below, is needed.

\subsection{Footprint Reduction}
\label{sec:footprint-reduction}

Traffic footprint in traffic engineering is defined as a product of traffic volume and the distance it travels within the network.
The distance is usually measured in terms of hops.
To calculate the footprint, we need to sum up all the products of every data packet sizes and their travel distances.
Footprint reduction is the fraction of reduction in footprint when caching is enabled (comparing to ``no caching'' case). Although we choose ``no caching'' as the baseline to calculate the footprint reduction, we must point out that the ranking of different strategies in an evaluation is irrelevant to the choice of the baseline strategy. The corresponding proof can be found in Appendix.

Compared to byte hit rate, footprint reduction takes into account \emph{where} the hit happens, since the number of hops is counted in the metric.
However, since footprint reduction uses the size of the content, it gives more accurate information about traffic reduction than simply using the average hops.
Also, it measures relative change and gives therefore a better picture of the differences between caching strategies.
Note that footprint reduction measures only reduction of intra-ISP traffic and does not give any indication about possible reductions in inter-ISP traffic.

Byte hit rate and footprint reduction are the two key metrics in evaluating performance of networks of ICN caches, but they must be used in conjunction; using only one of them leads to biased results (using only average hops will lead to even more bias).
For example, consider two caching strategies which achieve the same byte hit rate, but different footprint reductions.
Higher footprint reduction indicates that the hits happen closer to the clients, thus less intra-ISP traffic and generally better performance for the users.
We have identified and quantified the tradeoff between byte hit rate and footprint reduction~\cite{WangL:Cooperation} and will briefly outline this tradeoff below.

A naive solution for improving footprint reduction is to place the popular content as close to the clients as possible, i.e, edge caching~\cite{seyedicn2013}.
However, this leads to large redundancy in cached content and leads to (much) lower byte hit rate.
This tradeoff between the two metrics is mediated by a \emph{cooperation policy} which enables richer cooperation between the clients than a simple en-route caching allows.
As discussed in~\cite{WangL:Cooperation}, the tradeoff can be mediated by adjusting the number of copies for a content item and the range of how widely we search for the content in the network in case of a miss.
The search range covers all possible cases from en-route caching to searching the whole network (obviously with a cost that would need to be accounted for).
Adjusting the number of copies is harder to do exactly, but simple mechanisms like Cachedbit~\cite{wong:globecom2012} are likely to be sufficient in many cases.
For more details, we refer the reader to~\cite{WangL:Cooperation,wong:globecom2012}.

\subsection{Coupling Factor}
\label{sec:coupling-factor}

\begin{figure*}[!tb]
  \centering
  \subfloat[Strong coupling]{\includegraphics[width=5.5cm]{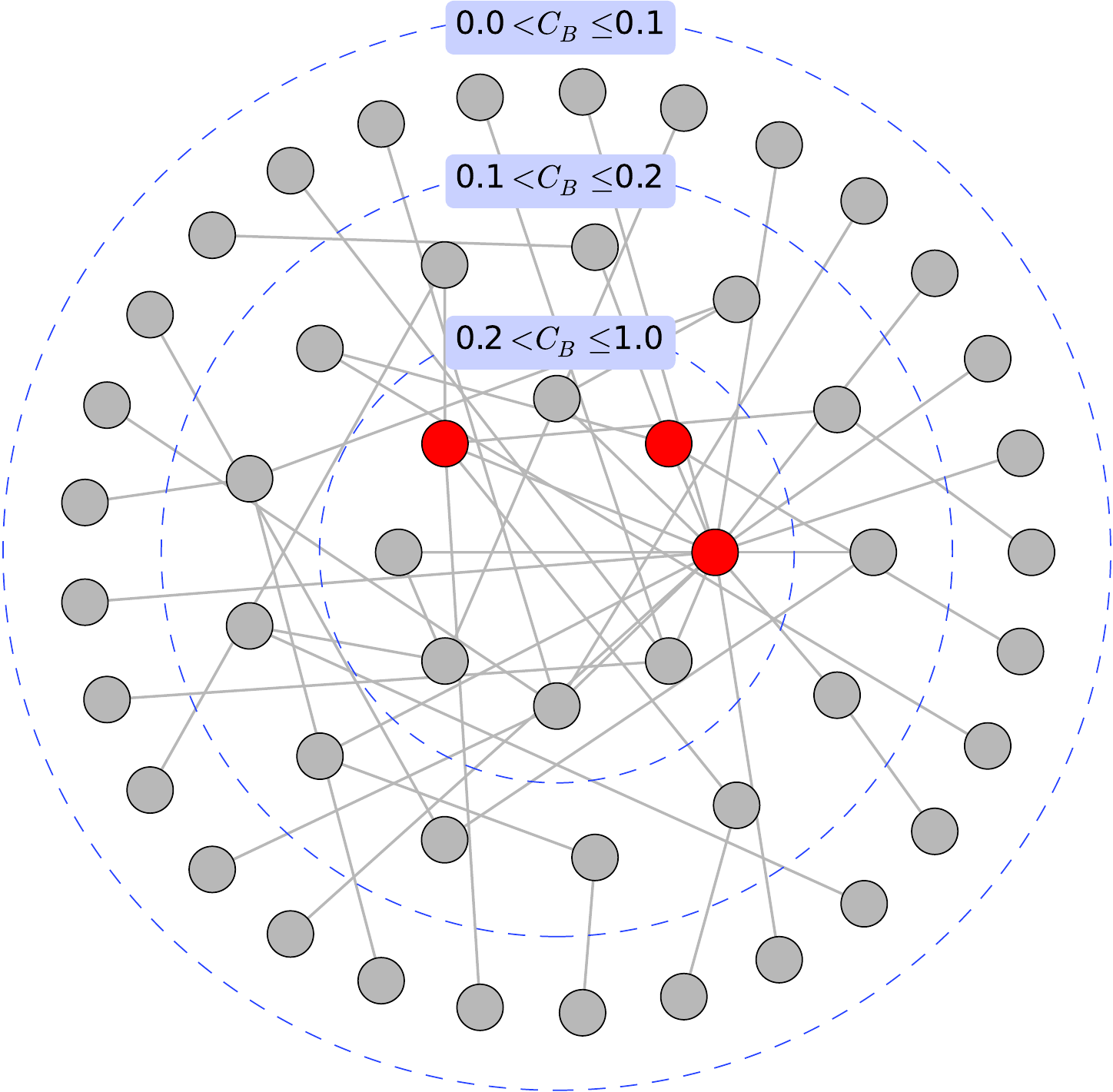}} 
  \quad
  \subfloat[Moderate coupling]{\includegraphics[width=5.5cm]{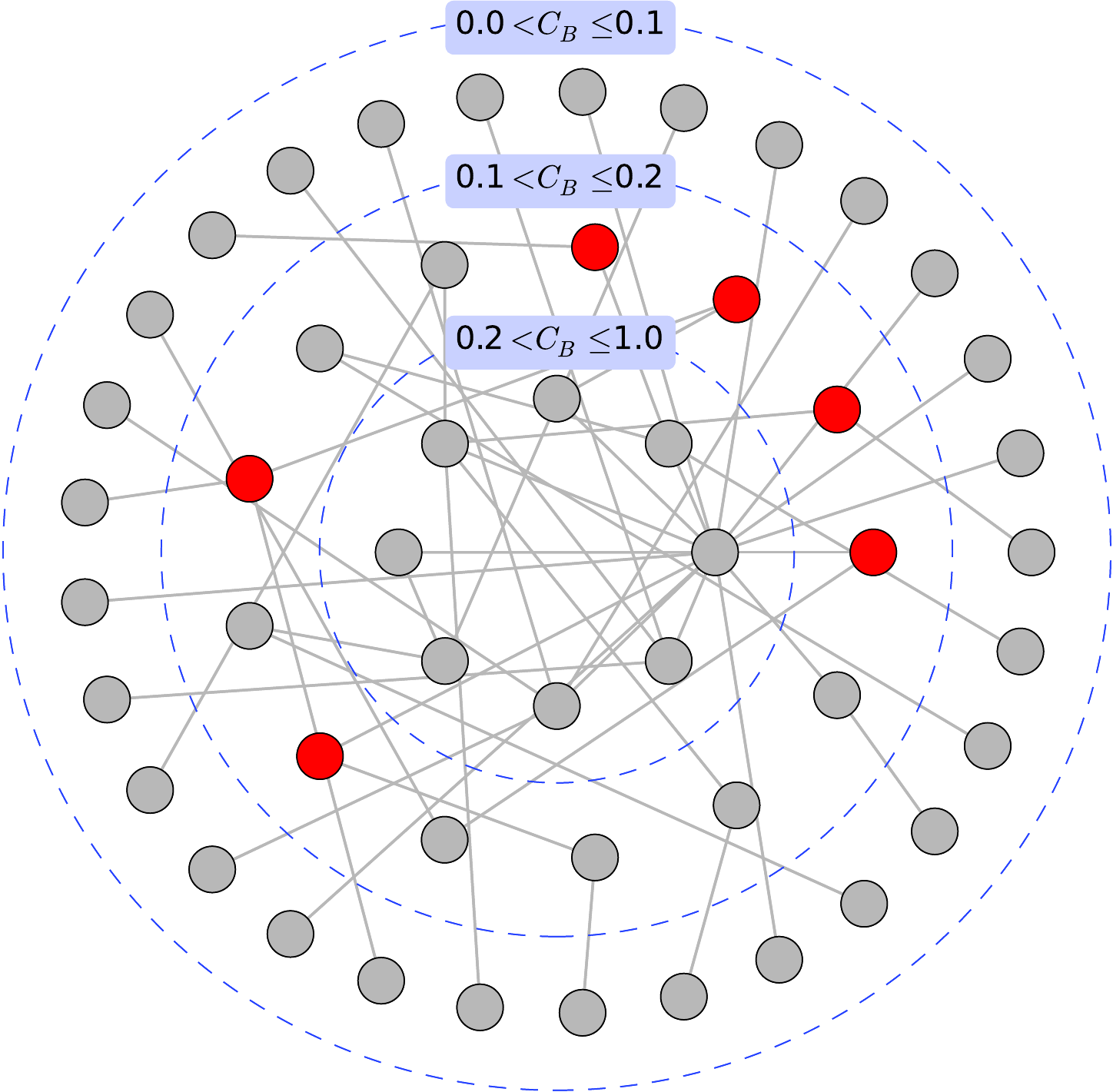}}
  \quad
  \subfloat[Weak coupling]{\includegraphics[width=5.5cm]{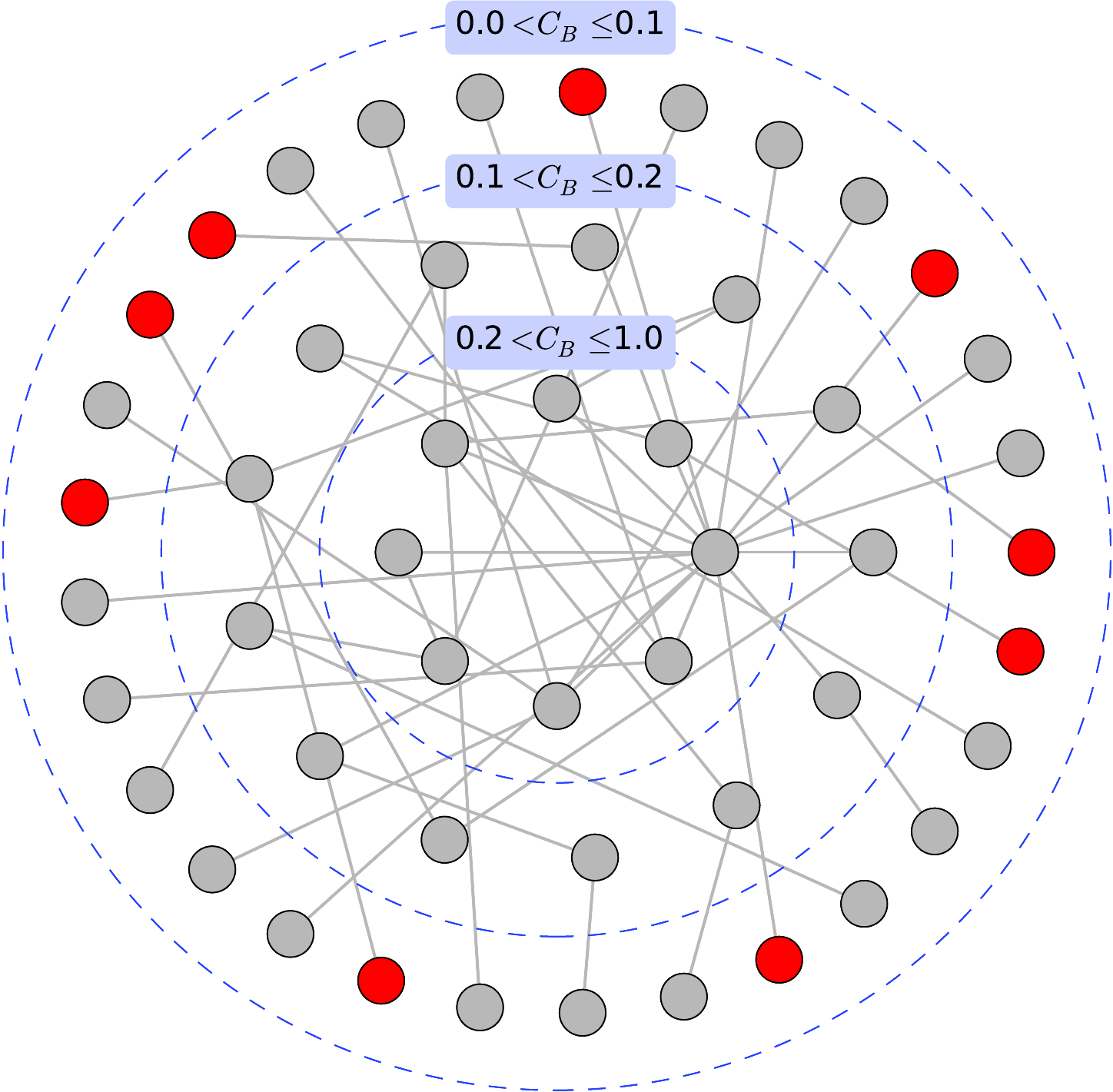}}
  \caption{Coupling between content popularity and network topology. The nodes are grouped with three concentric circles according to their betweenness centrality values $C_B$. Red color represents where the most popular content reside.}
  \label{fig:coupling}
\end{figure*}

We propose a new metric in this paper: the \emph{coupling factor}, to capture the effects of the network topology on the performance of caching.
We achieve this by identifying the ``position'' in the network where the hit happens.
Recall that byte hit rate does not give any information about where the hit happens, and footprint reduction is limited to finding content only along the routing path.
A cooperation policy that searches wider in the network is able to find content in other locations as well.
In this case, the position has a direct impact on the calculation of the metrics.


We define the coupling factor as a function of content popularity and network topology and it measures the impact of topology on content placement, and thus the impact on metrics like byte hit rate and footprint reduction.
Content popularity is easy to obtain, but for characterising topology, we have many more options, such as degree centrality, betweenness centrality, closeness centrality and so on.
Therefore, coupling factor can have several forms depending on which metrics are used in calculating the correlation, but the general idea is the same: we need a way of showing the relationship between popularity and topology.
For example, \cite{WangL:Cooperation} specifically chooses betweenness centrality and Pearson correlation in the calculation.

Figure~\ref{fig:coupling} shows how different degrees of coupling affect the placement of the most popular content.
The red dots represent the most popular content and the concentric circles group nodes according to their betweenness centrality.
Strong coupling means that the most popular content is placed in the nodes with high betweenness, i.e., the network core.
Weak coupling means the opposite, i.e., the popular content is placed at the network edge.
(Strictly speaking, if using correlation between popularity and node degree as a metric, strong coupling is indicative of strong positive correlation and weak coupling implies strong negative correlation.)
By adjusting the range of the search and the number of copies, we can influence the placement of content in the system, i.e., adjust the degree of coupling.
When the popular content is in the core, we improve byte hit rate and when the popular content is near the edge, we favor footprint reduction.
This means that the two parameters, search range and number of copies, can be used to adjust the tradeoff between the two metrics.

\section{Other Useful Metrics}
\label{sec:other}

The metrics we presented in Section~\ref{sec:metrics} are by no means the only metrics that can be used for measuring cache performance.


As is well known, the popularity distribution significantly influences caching performance in all kinds of caches.
However, new content is constantly added and popularity of content changes which may influence the metrics that are being used to measure the performance of caches.
Typically, some kind of aging is used to rid the system of old popularity information and test a caching strategy's ability to adapt to changes in content popularity.
Since popularity typically follows a power law and is characterised by a (single) parameter, it is a common technique to measure the effects of changes in that parameter on other system metrics.


Another interesting measure would be to find a way to quantify the filtering effects.
Filtering effect has been noticed in hierarchical cache systems~\cite{Williamson:2002:FEW}, and it refers to the phenomenon that the popularity of the most popular part in a miss stream is ``flattened'' because the downstream router filters out the most popular content.
Filtering effect degrades the caching performance of upstream routers.
We do not yet have a good metric to quantify the filtering effect, however results in~\cite{wong:globecom2012} seem to indicate that simple solutions might be able to fight the filtering effect, thus obviating the need for its measurement.
However, this is a question for future research.


As mentioned, content popularity changes dynamically, but also the network topology changes, be it due to failures or simply an evolution of the infrastructure.
This question needs more research in order to evaluate how often these kinds of topology changes affect an information-centric cache network and how large the impact would be.


\section{Beyond the Metrics: How to Run Experiments}
\label{sec:exp}

Choosing the right metrics is the first step, but it is not enough.
Designing experiments is as critical as selecting the right measurements, since it directly influences the ability to extract the right information from the experiment, as needed by the metrics that are to be computed.
A poorly designed experiment will not allow the correct information to be extracted, leading to possibly erroneous conclusions.
In terms of ICN experiments, there are three key elements in experiment design: \emph{content}, \emph{topology}, and \emph{traffic model}.


Content popularity is a key factor affecting performance of caching systems.
In the absence of publicly available request traces from recent years, many researchers are forced to use synthetic request traces.
For synthetic traces, it is important that its characteristics match those of realistic traces as closely as possible; obviously it will not be an exact match, but basic statistical characteristics should match real traces.
Research has shown that real world traces can often be modeled by either Zipf of Weibull distributions~\cite{Cha:2007:ITY:1298306.1298309}.
A further aspect of modeling the content is the size of the content set, i.e., how many objects are included in the synthetic trace; a real trace has its own fixed number of objects.
The reason why the content set size is important is the ``heavy tail'' nature of these popularity models which drags down performance is the content set is too large (for the amount of cache storage in the system).


Topology of the network is also of high importance in a correct experimental setting.
Research has showed that most realistic networks, including the Internet are scale-free, so a Barab\'{a}si-Albert model could be used to generate synthetic network topologies.
However, as work in~\cite{SpringN:Rocketfuel} shows, real ISP router topologies do not always conform to the scale-free model, so it is also important to experiment with real topologies.
As opposed to request traces, network topologies, both at ISP and AS level, are readily available, so unless the experiment setting requires a network much larger than exists in traces, using a real network topology is preferable to synthetic topologies.
Obviously, work needs to consider many different topologies in order to eliminate specifics of particular topologies from the results.


The third element, traffic, has two components: how much traffic and between which points it goes?
In the real world, requests from users at the network edge and traffic volume is typically proportional to the population of an edge routers geographical location.
These observations have led to the so-called ``gravity model'' which has been used in many papers.
Typically servers are connected to routers with high degree, i.e., near the core of the network and clients are connected to edge routers.
The question then becomes how many clients should the network have?
If all routers are assigned as servers or clients, there is no ``intermediate layer'' of caches and this may lead to a stronger filtering effect which degrades the performance of the system.
Based on our experience, we have observed that placing a small number of servers according to the gravity model and then assigning 20--30\% of the edge routers as clients works well.
In addition, the placement of clients should be randomized, varied from one experiment run to another, and we should perform a sufficient number of repetitions to guarantee statistically meaningful results, using for example confidence intervals to determine the number of repetitions.


\section{Lessons Learned}
\label{sec:lessons-learned}

We now summarize the key findings and observations in the paper.

\begin{itemize}
\item We need to evaluate the performance of the network as a whole as opposed to optimising individual caches.
\item Byte hit rate measures only reduction in inter-ISP traffic and due to content providers installing their servers close to clients may not reflect a true reduction in outgoing (costly) traffic.
\item Using average hops as a metric is not a good idea since it does not discriminate well and has issues related to assigning the number of hops for cache misses.
\item Footprint reduction measures reduction in intra-ISP traffic and gives more information about where the hits happen.
  It is to be preferred over byte hit rate and average hops.
\item A more refined metric, like the coupling factor, can help characterise the performance by looking at several aspects (content popularity and topology in this case) and provide insight into the inherent behavior of a network of caches.
\item Experiment setup needs careful thinking and control in order to get meaningful results.
\end{itemize}

\section{Conclusion}
\label{sec:conclusion}

In this article, we have argued that measuring the performance of caching in information-centric networks is fundamentally different from previous networks of caches, like web caching hierarchies.
We have discussed different metrics and showed common mistakes in use of metrics in caching in information-centric networks.
In addition, we have highlighted some lesser-known metrics which we consider more appropriate, and have proposed new metrics that capture more fine-grained information about the performance of an ICN caching network.
Finally, we have considered how ICN caching experiments should be set up and discussed key elements of experimental work.


\appendix
\begin{mythm}
In an evaluation, the ranking of caching strategies using footprint reduction metric is irrelevant to the choice of baseline strategy.
\end{mythm}

\begin{proof}
By definition, ``no caching'' strategy is selected as
a baseline in calculating footprint reduction metric. In other words, it
represents the traffic saving comparing to a special caching strategy
in which all the caches are disabled. Doubts have been raised that if
we change to another caching strategy as baseline, whether the results
(or observation) will remain valid. Below, we give a simple proof to
show it actually does not matter which strategy we choose as a
baseline as the theorem above states, the results are affine.

Let $x_{\alpha}$, $x_{\beta}$ and $x_{\theta}$ denote the traffic
footprint for caching strategy $\alpha$, caching strategy $\beta$, and
no caching $\theta$ respectively, and $y_{\alpha}$, $y_{\beta}$, and
$y_{\theta}$ are the corresponding footprint reduction. Our current
way of calculating footprint reduction is defined as:
\begin{align}
  & y_{\alpha} = 1 - \frac{x_{\alpha}}{x_{\theta}} \label{eq:1}\\
  & y_{\beta} = 1 - \frac{x_{\beta}}{x_{\theta}}  \label{eq:2} \\
  & y_{\theta} = 1 - \frac{x_{\theta}}{x_{\theta}} = 0 \quad
  \text{(baseline)} \label{eq:3}
\end{align}

As we can see, $x_{\theta}$ is just the baseline we are comparing
to. Obviously, the footprint reduction $y_{\theta}$ is zero when
comparing against itself. We can of course change the baseline to the
caching strategy $\beta$'s footprint $x_{\beta}$, then we have the
corresponding new metrics $y_{\alpha}'$, $y_{\beta}'$ and
$y_{\theta}'$ calculated as below:
\begin{align}
  & y_{\alpha}' = 1 - \frac{x_{\alpha}}{x_{\beta}} \label{eq:4} \\
  & y_{\beta}' = 1 - \frac{x_{\beta}}{x_{\beta}} = 0 \quad \text{(baseline)} \label{eq:5} \\
  & y_{\theta}' = 1 - \frac{x_{\theta}}{x_{\beta}} \label{eq:6}
\end{align}

From Eq (\ref{eq:2}), we have $x_{\beta} = x_{\theta}(1 -
y_{\beta})$. If we let $a = 1 - y_{\beta}$ and $b = 1 -
\frac{1}{1-y_{\beta}}$, then Eq (\ref{eq:4}), (\ref{eq:5}) and
(\ref{eq:6}) can be rewritten as follows

\begin{align}
  y_{\alpha}' & = 1 - \frac{x_{\alpha}}{x_{\beta}} = 1 - \frac{x_{\alpha}}{x_{\theta}(1 - y_{\beta})} \\ 
  & = \frac{1}{1-y_{\beta}} (1 - \frac{x_{\alpha}}{x_{\theta}}) + 1 - \frac{1}{1-y_{\beta}} = a y_{\alpha} + b \label{eq:7}
\end{align}

\begin{align}
  y_{\beta}' & = 1 - \frac{x_{\beta}}{x_{\beta}} = 1 - \frac{x_{\beta}}{x_{\theta}(1 - y_{\beta})} \\
  & = \frac{1}{1-y_{\beta}} (1 - \frac{x_{\beta}}{x_{\theta}}) + 1 - \frac{1}{1-y_{\beta}} = a y_{\beta} + b \label{eq:8}
\end{align}

\begin{align}
  y_{\theta}' & = 1 - \frac{x_{\theta}}{x_{\beta}} = 1 - \frac{x_{\theta}}{x_{\theta}(1 - y_{\beta})} \\
  & = \frac{1}{1-y_{\beta}} (1 - \frac{x_{\theta}}{x_{\theta}}) + 1 - \frac{1}{1-y_{\beta}} = a y_{\theta} + b \label{eq:9}
\end{align}
So we can see that the new metrics ($y_{\alpha}'$, $y_{\beta}'$ and
$y_{\theta}'$) are simply affine transformation of the old ones
($y_{\alpha}$, $y_{\beta}$ and $y_{\theta}$). It means the footprint
reduction is independent of the choice on which caching strategy as
baseline, since it will not change the ``ranking'' of the results.  In
some sense, geometrically, it simply means where we want to set our
``origin'' point, namely ``0'' point.
\end{proof}



\end{document}